\documentclass[10pt,twocolumn]{article}
\usepackage[letterpaper,margin=1in]{geometry}
\usepackage[T1]{fontenc}
\usepackage{lmodern}
\usepackage[hyphens]{url}
\usepackage{graphicx}
\usepackage{natbib}
\usepackage{caption}
\usepackage{booktabs}
\usepackage{amsmath,amssymb,amsthm,mathtools}
\usepackage{array}
\newcolumntype{L}[1]{>{\raggedright\arraybackslash}p{#1}}
\usepackage[hidelinks]{hyperref}
\newtheorem{definition}{Definition}
\newtheorem{theorem}{Theorem}

\newtheorem{corollary}{Corollary}
\newtheorem{proposition}{Proposition}

\newcommand{\Dist}{\mathcal{D}}
\newcommand{\supp}{\operatorname{supp}}
\newcommand{\TV}{\operatorname{TV}}
\newcommand{\GenOS}{\textnormal{\textsc{GenOS}}}
\newcommand{\Front}{\mathsf{Front}}
\newcommand{\Gen}{\mathsf{Gen}}
\newcommand{\Obs}{\mathsf{Obs}}
\newcommand{\VerifyOp}{\mathsf{Check}}
\newcommand{\Commit}{\mathsf{commit}}

\newcommand{\push}[1]{(#1)_{\#}}
\newcommand{\class}[1]{[#1]}
\newcommand{\liftrel}[1]{\overline{#1}}
\newcommand{\defect}{\operatorname{def}}
\newcommand{\bind}{\mathbin{;}}

\title{GenOS: Compositional Certificates for Semantic Robustness in AI Code Generation}
\author{Corrado Priami\\
Dipartimento di Informatica, Universit\`a di Pisa and Fondazione
Start Attractor\\
\texttt{corrado.priami@unipi.it}}
\date{}

\hypersetup{
  pdftitle={GenOS: Compositional Certificates for Semantic Robustness in AI Code Generation},
  pdfauthor={Corrado Priami},
  pdfsubject={Probabilistic operational semantics and compositional certificates for AI code-generation workflows},
  pdfkeywords={AI code generation, operational semantics, probabilistic programs, semantic robustness, certificates}
}

\begin{document}
\maketitle

\begin{abstract}
AI coding agents are stochastic programs: a prompt is interpreted, candidate artifacts are sampled, validators return observations, and an orchestrator commits or repairs. Small changes to a prompt or specification can therefore change a distribution of program behaviors even when the texts appear synonymous. Existing code-generation systems and benchmarks measure correctness, but do not provide a compositional criterion for when one prompt, contract, generator, or program may safely replace another inside a complete agentic workflow.

We introduce \GenOS{}, a probabilistic operational semantics for this replacement problem. Each layer is a Markov kernel and each semantic interface carries an observer-relative equivalence. Our main result proves that an equivalence-compatible kernel descends to quotient classes and that quotienting commutes with distributional extension and sequential composition. Consequently, equivalent prompts induce equal probabilities for every downstream equivalence-closed event, including verified commit. We further prove an operational workflow-bisimulation theorem, guarded-commit safety under sound validation, total-variation non-expansiveness, and an additive robustness bound that isolates approximation defects at individual pipeline layers.

An executable insertion-sort audit instantiates the theory with natural-language paraphrases, a formal contract, six concrete programs, two program observers, and exhaustive execution on 121 inputs. Equivalent prompts have identical code-class and commit distributions; a prompt with 5\% mass on an in-place contract is separated by a mutation observer while its downstream distances remain below the predicted bound. Twenty thousand randomized finite-kernel trials produce no violation of the exact or approximate laws. The framework is model-parametric: compatibility is a measurable property to test, not an assumption that language models automatically satisfy.
\end{abstract}

\section{Introduction}

A contemporary AI programming system is not a function from text to code. It is a workflow: interpret an instruction, retrieve context, sample one or more artifacts, execute tests or proof tools, repair failures, and decide whether to commit. Systems such as Jigsaw combine neural generation with program analysis and synthesis \cite{jain2022jigsaw}; Clover, AutoSpec, and VeCoGen place formal checking inside iterative generation loops \cite{sun2024clover,wen2024autospec,sevenhuijsen2024vecogen}; and prompt languages such as LMQL make model interaction and constraints programmable \cite{beurer2023lmql}. These systems expose a semantic problem that pass rates alone do not answer:

\begin{quote}
When can one component of a stochastic code-generation workflow be replaced by another without changing the behavior observable at the end of the workflow?
\end{quote}

Textual equality is too strong and informal synonymy is too weak. ``Return a sorted copy'' and ``produce a nondecreasing permutation without modifying the input'' may denote the same contract. Two generators may assign different probabilities to syntactically distinct programs but the same mass to behavioral classes. A test suite may identify programs that a mutation-sensitive observer separates. A verifier can preserve safety even when the generator is unreliable, but only if the commit guard is connected formally to target-language truth.

We propose \GenOS{} (Generative Operational Semantics), a layered semantics in which prompts, contracts, code artifacts, validator outcomes, and workflow observations are linked by probability kernels. At each interface, an equivalence states what the next layer is allowed to forget. The central proof obligation is \emph{compatibility}: equivalent inputs must induce equal mass on equivalent outputs. Compatibility yields a quotient kernel, and quotienting commutes with composition. This turns a local condition at each interface into an end-to-end semantic certificate.

\paragraph{Contributions.}
(1) We define an observer-relative kernel semantics for prompt-driven code generation, including a probabilistic labelled transition system (PLTS) for generate--check--repair--commit workflows.
(2) We prove the \emph{commuting quotient theorem}: compatible kernels induce well-defined quotient kernels, and quotienting commutes with Kleisli composition.
(3) We derive prompt-replacement and workflow-bisimulation corollaries, plus guarded-commit safety.
(4) We introduce a quotient total-variation metric and prove non-expansiveness and an additive bound for approximately compatible layers.
(5) We provide an executable finite audit with concrete Python programs, exhaustive observers, exact distributions, and randomized law checking.

The aim is not to claim that a particular model is invariant to paraphrase. FormalBench, for example, reports substantial instability under semantics-preserving transformations \cite{lecong2025formalbench}. \GenOS{} makes such instability a quantified interface defect and shows how it propagates.

\section{Related Work}

\paragraph{AI code generation and validation.}
HumanEval and MBPP established functional evaluation for code synthesis \cite{chen2021humaneval,austin2021mbpp}; EvalPlus demonstrates that weak test suites induce overly coarse notions of correctness \cite{liu2023evalplus}. In \GenOS{}, a test suite, prover, or analyzer is an observer whose discriminating power explicitly determines program equivalence. Jigsaw uses semantic post-processing \cite{jain2022jigsaw}, while Clover, AutoSpec, VeCoGen, and recent formal-verification pipelines use generated annotations, proof obligations, and repair feedback \cite{sun2024clover,wen2024autospec,sevenhuijsen2024vecogen,astrogator2025}. These systems motivate guarded workflows, but they do not state conditions under which equivalence at one pipeline layer is preserved by all subsequent layers.

\paragraph{Prompts and formal specifications.}
LMQL gives prompts control flow and output constraints \cite{beurer2023lmql}. nl2spec maps natural language to temporal logic while exposing ambiguity for interactive correction \cite{cosler2023nl2spec}. FormalBench evaluates specification consistency, completeness, and robustness \cite{lecong2025formalbench}. \GenOS{} is complementary: it treats prompt interpretation and formalization as kernels, then relates their quotient behavior to generated programs and commit decisions.

\paragraph{Operational and probabilistic semantics.}
Our transition-based view follows structural operational semantics \cite{plotkin2004sos}; the lifting of equivalence to distributions follows probabilistic bisimulation and probabilistic automata \cite{larsen1991bisimulation,segala1995probabilistic,baier2008principles}. Pi-calculus is structurally appropriate for communicating agents \cite{milner1999communicating}, but classical stochastic pi-calculus places rates on actions \cite{priami1995stochastic}. Model calls instead expose uncertainty over \emph{which artifact} is produced, motivating transitions to distributions over artifacts and successor configurations.

\section{Layered Generative Operational Semantics}

We use finite distributions; countable and measurable versions follow with standard kernel assumptions. For a set $X$, $\Dist(X)$ is the set of finitely supported probability distributions on $X$, and $\delta_x$ is the point distribution at $x$. A kernel $K:X\to\Dist(Y)$ has affine extension
\[
K^{\dagger}(\mu)(y)=\sum_x \mu(x)K(x)(y).
\]
For $K:X\to\Dist(Y)$ and $L:Y\to\Dist(Z)$, sequential composition is $(K\bind L)(x)=L^{\dagger}(K(x))$.

A minimal pipeline has prompt sources $P$, contracts $S$, generated artifacts $C$, and observations $O$:
\[
P \xrightarrow{\Front} \Dist(S)
\xrightarrow{\Gen} \Dist(C)
\xrightarrow{\Obs} \Dist(O).
\]
When validation depends on the contract, generation preserves the contract tag: $\widehat{\Gen}(s)$ is a distribution on $S\times C$ with first component $s$, and $\Obs:S\times C\to\Dist(O)$. Observations may record tests, proof results, diagnostics, repair requests, resource use, mutation, exceptions, or commit.

Table~\ref{tab:interfaces} summarizes the proof obligation at each interface. The equivalences are not required to be identical across applications: a security audit can refine the code observer with information-flow events, while a cost audit can add time or token use to workflow observations.

\begin{table}[t]
\centering
\footnotesize
\setlength{\tabcolsep}{4pt}
\begin{tabular}{L{1.45cm}L{1.65cm}L{2.45cm}}
\toprule
Layer & Typical equivalence & Local audit obligation \\
\midrule
Intent & Same contract mass & Paraphrases formalize alike \\
Generate & Same behavior mass & Syntax preserves behavior \\
Validation & Same verdict class & Tools distinguish\newline selected classes \\
Control & Same abstract trace & Replacement is safe in context \\
\bottomrule
\end{tabular}
\caption{Semantic interfaces and their local compatibility obligations.}
\label{tab:interfaces}
\end{table}

The same kernels generate a PLTS. A configuration contains a control location, semantic value, context, and provenance trace. A model call steps to a distribution; deterministic parsing, communication, and commit rules step to point distributions. For example,
\begin{multline*}
\langle \mathsf{generate}(s),\Gamma,\kappa\rangle
 \xrightarrow{\mathsf{sample}} {}\\
 \sum_c \Gen(s)(c)\,\delta_{\langle \mathsf{check}(s,c),\Gamma,\kappa c\rangle}.
\end{multline*}
Thus the kernel and operational views describe the same transition structure at different granularities.

\paragraph{Observer-relative equivalence.}
Let $R_X$ be an equivalence on $X$ and $q_X:X\to X/R_X$ its quotient map. The lifting to distributions compares probability mass on equivalence classes.

\begin{definition}[Lifted equivalence]
For $\mu,\nu\in\Dist(X)$, write $\mu\,\liftrel{R_X}\,\nu$ iff, for every class $E\in X/R_X$, $\mu(E)=\nu(E)$. Equivalently, $\push{q_X}\mu=\push{q_X}\nu$.
\end{definition}

Program equivalence is deliberately parameterized by an observer. A return-value observer may equate a pure sort with an in-place sort; an observer that also records the post-state separates them. Contract equivalence can be defined extensionally by satisfaction over program classes, and prompt equivalence by the contract-class mass induced by $\Front$.

\begin{definition}[Compatible kernel]
A kernel $K:X\to\Dist(Y)$ is $(R_X,R_Y)$-compatible when $xR_Xx'$ implies $K(x)\,\liftrel{R_Y}\,K(x')$.
\end{definition}

Compatibility is the semantic substitutability test for a layer. It is inspectable empirically by repeated sampling and classifying outcomes, or proved for deterministic front ends and symbolic transformations.

\paragraph{Contracts and extensional validators.}
Let $\models\;\subseteq S\times C$ be target-language satisfaction. A useful contract relation is
\[
sR_Ss'\quad\Longleftrightarrow\quad
\forall c,c'.\ cR_Cc'\Rightarrow
(s\models c\Leftrightarrow s'\models c').
\]
This definition says that equivalent contracts select the same accepted code classes, not necessarily the same syntax or proof obligations. A deterministic validator is \emph{extensional} when its verdict and diagnostic class depend only on $\class{s}$ and $\class{c}$.

\begin{proposition}[Validator compatibility]\label{prop:validator}
If contract satisfaction is closed under $(R_S,R_C)$ and a validator is extensional, then the tagged validation kernel on $S\times C$ is compatible with the product equivalence. If its accept verdict is also sound, the resulting commit observation is both equivalence-preserving and safe.
\end{proposition}

The proposition distinguishes two properties that are often conflated. Extensionality states that equivalent inputs receive equivalent observations; soundness states that acceptance implies semantic truth. Either can hold without the other.

\section{Compositional Certificates}

\begin{theorem}[Quotient kernel and commuting square]\label{thm:quotient}
If $K:X\to\Dist(Y)$ is $(R_X,R_Y)$-compatible, then
\[
\overline K(\class{x})=\push{q_Y}(K(x))
\]
defines a unique kernel $\overline K:X/R_X\to\Dist(Y/R_Y)$. Moreover, for every $\mu\in\Dist(X)$,
\[
\push{q_Y}\bigl(K^{\dagger}(\mu)\bigr)
=\overline K^{\dagger}\bigl(\push{q_X}\mu\bigr).
\]
\end{theorem}

\begin{proof}
Compatibility makes the definition independent of the representative $x$. For a class $B\in Y/R_Y$, both sides equal
$\sum_x\mu(x)K(x)(q_Y^{-1}(B))$ after grouping the outer sum by $R_X$-classes. Uniqueness follows because point distributions on quotient classes determine a kernel.
\end{proof}

\begin{theorem}[Quotients commute with composition]\label{thm:composition}
Let $K:X\to\Dist(Y)$ and $L:Y\to\Dist(Z)$ be compatible with $(R_X,R_Y)$ and $(R_Y,R_Z)$, respectively. Then $K\bind L$ is $(R_X,R_Z)$-compatible and
\[
\overline{K\bind L}=\overline K\bind\overline L.
\]
\end{theorem}

\begin{proof}
Apply Theorem~\ref{thm:quotient} first to $L^{\dagger}(K(x))$, then to $K(x)$. The resulting pushforward is $\overline L^{\dagger}(\overline K(\class{x}))$.
\end{proof}

\begin{corollary}[End-to-end prompt replacement]\label{cor:replacement}
If all layers of a \GenOS{} pipeline are compatible and prompts $p,p'$ induce equivalent contract distributions, then the final observation distributions agree on every $R_O$-class. Hence every equivalence-closed event $A\subseteq O$, including a commit event, has equal probability under $p$ and $p'$.
\end{corollary}

This result separates syntax-level variation from observable semantic variation. Artifact probabilities may differ substantially; only class mass must be preserved.

\paragraph{Operational consequence.}
Relate two workflow configurations when they are at the same control location, their current values are equivalent, and their traces agree after quotienting recorded artifacts and observations. Each deterministic rule is matched by a point transition, and each stochastic rule is matched by compatibility.

\begin{theorem}[Workflow bisimulation]\label{thm:bisim}
If every stochastic action kernel in two structurally identical \GenOS{} workflows is compatible and deterministic rules respect the selected equivalences, the induced configuration relation is a probabilistic bisimulation. Consequently, the workflows assign equal probability to every equivalence-closed finite trace cylinder.
\end{theorem}

The proof is by cases on the operational rules. The theorem supports local replacement inside loops and agent networks, not only a fixed acyclic pipeline.

\paragraph{Guarded commitment.}
Let $s\models c$ mean that target artifact $c$ satisfies contract $s$. A checker is sound when $\VerifyOp(s,c)=\mathsf{accept}$ implies $s\models c$. The commit rule is guarded when only accepted pairs can reach $\Commit(c)$.

\begin{theorem}[Safe commit]\label{thm:safe}
In a guarded workflow with a sound checker, every reachable committed artifact satisfies its associated contract, independent of the generator distribution and repair policy.
\end{theorem}

\begin{proof}
The final transition into a committed state requires a preceding accept observation. Soundness converts this observation into $s\models c$; earlier stochastic choices are irrelevant.
\end{proof}

The theorem does not claim completeness: correct programs may be rejected. It also exposes a key boundary. Tests or learned judges usually justify only a weaker, probabilistic soundness claim, whereas proof-producing or sound static tools can justify the premise directly.

\section{Repair, Iteration, and Provenance}

A core workflow grammar is sufficient to expose iterative behavior:
\begin{align*}
W ::= {} & \mathsf{generate}(s)\mid\mathsf{check}(s,c)
\mid\mathsf{repair}(s,c,d)\\
&\mid\mathsf{commit}(s,c)\mid W;W\mid\mathsf{repeat}(W).
\end{align*}
The configuration trace records events such as generated artifacts, checker results, diagnostics, and commits. Rules only append events.

\begin{proposition}[Trace monotonicity]\label{prop:trace}
If $\langle W,\Gamma,\kappa\rangle\rightarrow\mu$, every configuration in $\supp(\mu)$ has a trace with prefix $\kappa$. Therefore provenance is monotone along every finite execution.
\end{proposition}

Trace monotonicity makes a compatibility claim auditable: the class labels and decisions used to justify replacement can be reconstructed from the run. Traces may be quotiented for privacy, provided the retained abstraction still contains the observations named by the certificate.

Repair is a new model call conditioned on the rejected artifact and diagnostic. Let $A_n$ be the event that attempt $n$ is accepted, conditional on all earlier failures. A useful progress theorem requires a lower bound on conditional success, but not independence.

\begin{theorem}[Conditional repair progress]\label{thm:repair}
Suppose a sound guarded workflow retries after rejection and, for every history with no earlier acceptance,
$\Pr(A_n\mid\neg A_1,\ldots,\neg A_{n-1})\ge p>0$.
Then
\[
\Pr(\text{no acceptance in the first }n\text{ attempts})
\le(1-p)^n.
\]
Thus acceptance occurs almost surely if retries continue indefinitely; every artifact eventually committed is correct by Theorem~\ref{thm:safe}.
\end{theorem}

\begin{proof}
By the chain rule, the probability of $n$ consecutive failures is the product of their conditional probabilities. Each is at most $1-p$, giving the bound. No independence premise is used.
\end{proof}

This theorem identifies what an empirical repair study should estimate: a history-uniform or history-stratified lower bound, rather than an unconditional average success rate that can hide repeated failure modes.

\section{Approximate Robustness}

Exact class-mass equality is often too strict for sampled models. Define the quotient total-variation distance
\begin{align*}
d_{R_X}(\mu,\nu)&=\TV(\push{q_X}\mu,\push{q_X}\nu),\\
\TV(\alpha,\beta)&=\tfrac12\sum_u|\alpha(u)-\beta(u)|.
\end{align*}

\begin{theorem}[Non-expansiveness]\label{thm:nonexpansive}
If $K$ is $(R_X,R_Y)$-compatible, then
\[
d_{R_Y}(K^{\dagger}\mu,K^{\dagger}\nu)\le d_{R_X}(\mu,\nu).
\]
\end{theorem}

\begin{proof}
By Theorem~\ref{thm:quotient}, both outputs are obtained by applying the same quotient kernel $\overline K$ to the quotient inputs. Total variation contracts under every Markov kernel.
\end{proof}

To diagnose imperfect layers, define their compatibility defect
\[
\defect(K)=\sup_{xR_Xx'} d_{R_Y}(K(x),K(x')).
\]
A defect of zero recovers exact compatibility. The next theorem turns interface defects into an end-to-end error budget.

\begin{theorem}[Additive robustness budget]\label{thm:additive}
For any kernel $K:X\to\Dist(Y)$,
\[
d_{R_Y}(K^{\dagger}\mu,K^{\dagger}\nu)
\le d_{R_X}(\mu,\nu)+\defect(K).
\]
For a chain $K_1\bind\cdots\bind K_n$, the final distance is at most the initial quotient distance plus $\sum_i\defect(K_i)$.
\end{theorem}

\begin{proof}[Proof sketch]
Match the common mass that $\push{q_X}\mu$ and $\push{q_X}\nu$ assign to each input class, coupling matched samples within that class. Their output quotient distributions differ by at most $\defect(K)$. Unmatched quotient mass is at most $d_{R_X}(\mu,\nu)$ and can contribute at most that amount. Iteration gives the chain bound.
\end{proof}

The budget is actionable: prompt interpretation, retrieval, generation, validation, and repair can be audited separately. A large downstream drift must be explained by an upstream prompt distance or by one or more measured compatibility defects.

\section{Estimating Compatibility from Samples}

For a black-box model, $K(x)$ is estimated by repeated calls under a frozen model version, context, decoding policy, and tool environment. Let the output quotient have $m$ classes and let $\widehat K_N(x)$ be the empirical class distribution from $N$ independent samples.

\begin{theorem}[Finite-sample certificate]\label{thm:sampling}
For a fixed input $x$ and $0<\delta<1$, with probability at least $1-\delta$,
\[
\max_E|\widehat K_N(x)(E)-K(x)(E)|
\le \sqrt{\frac{\log(2m/\delta)}{2N}}.
\]
For $M$ tested input pairs, let $\widehat d_j$ and $d_j$ be their empirical and population quotient total variations. If
\[
N\ge\frac{m^2}{2\tau^2}
\log\frac{4mM}{\delta},
\]
then $|\widehat d_j-d_j|\le\tau$ simultaneously for every pair, with probability at least $1-\delta$.
\end{theorem}

\begin{proof}
Apply Hoeffding to each class frequency and union-bound over $m$ classes, two endpoints, and $M$ pairs. With coordinate error at most $\tau/m$, each endpoint has $\ell_1$ error at most $\tau$; the reverse triangle inequality for $\ell_1$ gives total-variation error at most $\tau$.
\end{proof}

The independence condition concerns repeated sampling, not repair attempts in Theorem~\ref{thm:repair}. If an API changes model versions or uses stateful caching, the audit must stratify runs or use a dependence-aware confidence method. For an observed empirical defect $\widehat\eta$, Theorem~\ref{thm:sampling} supplies an upper confidence value $\eta^{+}$; substituting $\eta^{+}$ into Theorem~\ref{thm:additive} yields a high-probability end-to-end certificate.

A practical audit has five steps. First, declare the observational properties and resulting equivalence classes. Second, freeze or record the model, tools, decoding parameters, and context. Third, sample each proposed replacement pair and classify outputs. Fourth, estimate class distances and confidence bounds. Fifth, compose the upper bounds across the workflow and separately verify the soundness premise of any commit guard. A failed compatibility test is informative: its distinguishing class identifies the semantic behavior changed by the replacement.

\section{A Checkable Certificate Object}

The previous results can be packaged as a finite object that accompanies a proposed prompt, model, tool, or agent replacement. For a pipeline with interfaces $0,\ldots,n$, define a certificate
\[
\mathcal C=(R_0,\ldots,R_n;\epsilon_0;\eta_1^+,\ldots,\eta_n^+;
\delta;\mathcal G),
\]
where $R_i$ declares the observer at interface $i$, $\epsilon_0$ bounds the initial quotient distance, $\eta_i^+$ is a proved or high-confidence upper bound on layer $i$'s compatibility defect, $\delta$ is the joint statistical failure probability, and $\mathcal G$ records the justification for each commit guard. A guard entry may be \emph{sound}, may provide a false-accept upper bound $\beta$, or may be explicitly unsupported. The certificate is independently checkable from class definitions, samples or proofs, and trace data; it need not reveal prompts or generated source if privacy-preserving class evidence is sufficient.

\begin{theorem}[Certificate soundness]\label{thm:certificate}
Suppose every numerical assertion in $\mathcal C$ holds simultaneously with probability at least $1-\delta$. Let $A$ be any final event closed under $R_n$. For the original and replacement workflows,
\[
\bigl|\Pr_W(A)-\Pr_{W'}(A)\bigr|
\le \tau:=\min\{1,\epsilon_0+\textstyle\sum_{i=1}^n\eta_i^+\}
\]
with probability at least $1-\delta$. If a commit guard is sound, its unsafe-commit probability is zero. If instead its false-accept event has probability at most $\beta$, unsafe commit has probability at most $\beta$.
\end{theorem}

\begin{proof}
The event bound follows from Theorem~\ref{thm:additive}, because total variation upper-bounds the probability difference of every measurable quotient event. The confidence statement is inherited from the simultaneous bounds. The guard clauses follow from Theorem~\ref{thm:safe}, or by inclusion of unsafe commit in the false-accept event.
\end{proof}

The theorem clarifies three different outcomes. A \emph{robust and safe} certificate has small $\tau$ and a sound guard. A \emph{robust but uncertified} replacement has small $\tau$ but no checker-soundness argument: it behaves like the baseline but may reproduce its errors. A \emph{correct but non-robust} system may commit only verified programs while paraphrases yield very different rejection, repair, latency, or cost traces. These dimensions should be reported separately.

Certificate composition is modular. Two adjacent certificates can be concatenated when they use the same intermediate equivalence; their defect bounds add and their statistical failure probabilities combine by a union bound. Observer refinement is also checkable: Proposition~\ref{prop:observer} permits projection from a stronger certificate to a weaker observation boundary, but never the reverse. Consequently, a repository can maintain certificates at different assurance levels---for example, functional, frame-safe, resource-aware, and information-flow-aware---without treating them as interchangeable.

\section{Sharpness and Failure Modes}

The hypotheses above are not merely sufficient proof conveniences. Small finite counterexamples show why an audit must expose each of them.

\paragraph{Compatibility is necessary.}
Let $xR_Xx'$ and let $a,b$ be distinct $R_Y$-classes. If $K(x)=\delta_a$ and $K(x')=\delta_b$, then $\overline K(\class{x})$ would have to be both $\delta_{\class a}$ and $\delta_{\class b}$. No quotient kernel exists. In a code agent, this is exactly a semantics-preserving paraphrase that systematically changes the generated behavior class; composition cannot repair the missing local congruence.

\paragraph{Events must be quotient-closed.}
Suppose $aR_Yb$ and two pipelines output $\delta_a$ and $\delta_b$. They are equivalent at the declared interface, yet the concrete event $\{a\}$ has probabilities one and zero. Corollary~\ref{cor:replacement} therefore cannot promise equality for observations that the equivalence intentionally erased. This is why source-string equality, exact diagnostic wording, or a particular proof term cannot be queried from a certificate that preserves only functional behavior.

\paragraph{Extensionality and soundness differ.}
A checker that accepts every program is perfectly extensional under any program equivalence, but it is not sound. Conversely, a sound checker could expose irrelevant syntactic diagnostics and thereby fail extensionality. The first breaks safe commitment; the second breaks replacement while preserving correctness. Proposition~\ref{prop:validator} requires both properties only when both conclusions are desired.

\paragraph{Hidden state must enter the configuration.}
If a model service changes behavior after previous calls, a purported map from prompt to output distribution is not a kernel on prompts alone. Adding model version, decoding state, retrieved context, tool state, and relevant trace history to $\Gamma$ restores the Markov property. Omitting them can create spurious compatibility: two prompts agree in one call order but diverge after different histories.

\paragraph{The quantitative constants are tight.}
Non-expansiveness is tight for an identity quotient kernel, and a layer defect can contribute its full value when representatives attain $\defect(K)$. The additive budget is therefore a worst-case guarantee; tighter bounds require additional structure.

\section{Executable Semantic Audit}

We instantiate the definitions on insertion sort. The purpose is to exercise every semantic object exactly, not to estimate a production model's quality. The audit implementation uses only the Python standard library and emits all reported numbers.

\paragraph{Observers and programs.}
Six concrete functions implement: built-in sorted copy, insertion-sort copy, built-in in-place sort, descending sort, duplicate-removing sort, and an off-by-one insertion sort. We execute each on all 121 arrays of length $0$--$4$ over $\{-1,0,1\}$. A weak observer records return values; a strong observer also records the post-state of the input. Under the weak observer, the two pure implementations and the in-place implementation are equivalent. Under the strong observer, the in-place implementation is separated, while the two pure implementations remain equivalent. This provides a concrete distinguishing context rather than an asserted semantic difference.

\paragraph{Prompts and kernels.}
The prompt texts are: $p_A$, ``return a sorted copy and leave the input unchanged''; $p_B$, ``produce a nondecreasing permutation without mutating the argument''; and $p_F$, a formal pre/postcondition specifying nondecreasing order, multiset equality, and frame preservation. They induce different distributions over two syntactically distinct but equivalent copy-sort contracts. Their quotient contract distribution is the point mass on \emph{copy-sort}. The two contract-specific generation rows differ at artifact level but assign the same masses to five program-property classes: pure-correct $0.60$, mutating-correct $0.15$, descending $0.10$, lossy $0.10$, and partial-sort $0.05$. A sound copy-sort checker commits only pure-correct programs.

\begin{table}[t]
\centering
\footnotesize
\setlength{\tabcolsep}{3pt}
\begin{tabular}{lccc}
\toprule
Comparison & Contract TV & Code TV & Verdict TV \\
\midrule
$p_A$ vs. $p_B$ & 0 & $<1.5\cdot10^{-16}$ & $<6\cdot10^{-17}$ \\
$p_A$ vs. $p_F$ & 0 & $<8\cdot10^{-17}$ & $<3\cdot10^{-17}$ \\
$p_A$ vs. $p_{\mathrm{near}}$ & .0500 & .0275 & .0050 \\
\bottomrule
\end{tabular}
\caption{Quotient total-variation distances. Tiny nonzero values are floating-point roundoff.}
\label{tab:distances}
\end{table}

Table~\ref{tab:classes} makes the role of the observer explicit. The checker is stricter than the weak return-value observer because the copy contract contains a frame condition.

\begin{table}[t]
\centering
\footnotesize
\setlength{\tabcolsep}{3pt}
\begin{tabular}{lccc}
\toprule
Implementation & Weak class & Strong class & Copy commit \\
\midrule
Built-in copy & $E_1$ & $F_1$ & yes \\
Insertion copy & $E_1$ & $F_1$ & yes \\
In-place sort & $E_1$ & $F_2$ & no \\
Descending & $E_2$ & $F_3$ & no \\
Deduplicate & $E_3$ & $F_4$ & no \\
Off-by-one & $E_4$ & $F_5$ & no \\
\bottomrule
\end{tabular}
\caption{Exhaustive equivalence classes for six concrete programs.}
\label{tab:classes}
\end{table}

The three equivalent prompts all yield commit probability $0.60$ and identical rejection-class distributions, as predicted by Corollary~\ref{cor:replacement}. A near prompt places $0.05$ mass on an in-place contract. The strong observer detects the new mutation behavior. Nevertheless, code-class distance is $0.0275$ and verdict distance is $0.0050$, both below the input distance $0.05$ as required by Theorem~\ref{thm:nonexpansive}. A fully in-place prompt yields mutation-class mass $0.70$ and is not declared equivalent.

\paragraph{Randomized law checks.}
We additionally sample 10,000 pairs of distributions and random finite kernels of dimensions two through eight. No trial violates total-variation contraction; the largest observed output/input ratio is $0.9471$ (mean $0.3435$). A second 10,000-trial test checks
$\TV(\mu K,\nu L)\le\TV(\mu,\nu)+\sup_x\TV(K_x,L_x)$; again there are no violations. These tests do not replace proofs. They validate the implementation, numerical conventions, and artifact tables, and make the reported case study independently reproducible.

\section{Observer Design and Refinement}

A robustness claim is meaningful only relative to an observation boundary. Let an observer be a map $h:C\to B$ from concrete executions to observable records. It induces $cR_hc'$ exactly when $h(c)=h(c')$. Return values, mutated state, exceptions, proof certificates, resource use, information-flow events, and tool calls can be included independently. This makes the abstraction explicit and reviewable rather than hidden inside a benchmark metric.

\begin{definition}[Observer refinement]
Observer $h_2:C\to B_2$ \emph{refines} $h_1:C\to B_1$ when there exists a map $r:B_2\to B_1$ such that $h_1=r\circ h_2$. Thus $h_2$ retains at least the information exposed by $h_1$.
\end{definition}

\begin{proposition}[Monotonicity under observer refinement]\label{prop:observer}
If $h_2$ refines $h_1$, then $R_{h_2}\subseteq R_{h_1}$. Consequently, exact equivalence under $h_2$ implies equivalence under $h_1$, but not conversely. For distributions,
\[
 d_{R_{h_1}}(\mu,\nu)\le d_{R_{h_2}}(\mu,\nu).
\]
\end{proposition}

\begin{proof}
If $h_2(c)=h_2(c')$, applying $r$ gives $h_1(c)=h_1(c')$. The quotient for $h_1$ is therefore a further coarsening of the quotient for $h_2$, and total variation cannot increase under the induced deterministic map.
\end{proof}

This proposition gives a simple discipline for safety-critical audits: begin with the deployment observations and refine them until every property that can invalidate substitution is represented. A unit-test observer may equate two programs that agree on sampled inputs; adding exhaustive bounded executions, frame conditions, exceptions, or a sound verifier can split that class. A certificate produced under a coarse observer remains valid only for the coarse question. The insertion-sort audit below demonstrates the issue: return-only observation equates copying and mutating sorts, while the post-state observer separates them.

\section{Nondeterminism and Agent Control}

Real agents contain both probabilistic model calls and nondeterministic choices: a scheduler selects an agent, an orchestrator selects a tool, or a human chooses whether to retry. A fixed orchestration policy resolves these choices and induces the kernels used above. With unresolved choices, the semantics is a probabilistic automaton or Markov decision process rather than a Markov chain \cite{segala1995probabilistic,baier2008principles}.

There are two useful certificate strengths. A \emph{policy-specific} certificate proves compatibility for the deployed scheduler and tool-selection policy. A \emph{policy-robust} certificate requires every enabled action to be matchable by an action with equivalent successor-class mass; this is the standard alternating form of probabilistic bisimulation. The latter implies replacement under every scheduler that respects the action interface. Event probabilities can then be bounded by taking the supremum or infimum over schedulers. This distinction matters for agentic code generation: a compatible generator does not compensate for an orchestrator that routes semantically equivalent diagnostics to different repair tools.

For bounded repair horizons, Theorem~\ref{thm:additive} gives a conservative drift bound. If the initial prompt distance is $\epsilon$, and each of at most $T$ generation, validation, or repair actions has defect at most $\eta_t$, every quotient-trace event differs by at most $\epsilon+\sum_{t=1}^{T}\eta_t$. Under exact action compatibility, Theorem~\ref{thm:bisim} removes the dependence on $T$ entirely: loops and interleavings preserve equivalence because the relation is an invariant of every transition, not a one-pass approximation.

\section{Protocol for a Model-Based Evaluation}

The finite audit establishes that the definitions, theorems, and implementation agree. A full model evaluation should test whether contemporary code models satisfy the local obligations. Table~\ref{tab:protocol} defines an executable protocol that avoids conflating textual paraphrase invariance, functional correctness, and safe commitment.

\begin{table}[t]
\centering
\footnotesize
\begin{tabular}{L{1.20cm}L{1.75cm}L{2.60cm}}
\toprule
Stage & Construction & Reported quantity \\
\midrule
Tasks & HumanEval, MBPP, formal tasks & Task, language, contract, observer \\
Prompts & Paraphrases, formalizations, frame variants & Contract-class input distance \\
Samples & $N$ frozen calls per prompt and model & Empirical code-class distributions \\
Checking & EvalPlus tests and sound tools where available & Verdict classes and commit events \\
Analysis & Pairwise and joint bounds & Defects, distinguishing classes, budget \\
Repair & Fixed diagnostic templates and retry cap & Conditional success by failure history \\
\bottomrule
\end{tabular}
\caption{Protocol for testing a model-mediated \GenOS{} pipeline.}
\label{tab:protocol}
\end{table}

\paragraph{Prompt families.}
Each task should contain: (i) meaning-preserving lexical and syntactic paraphrases; (ii) natural-language and formal or semi-formal versions of the same contract; and (iii) deliberately non-equivalent near misses, such as changing an in-place requirement, exception policy, numerical precondition, or complexity constraint. The third group is essential: a metric that declares both paraphrases and near misses equivalent has selected an inadequate observer.

\paragraph{Artifact classification.}
Programs are first compiled and executed in a sandbox. Functional observations can combine reference tests with augmented suites such as EvalPlus \cite{liu2023evalplus}; frame, exception, and resource observations are recorded separately. When a formal specification and a sound tool are available, proof acceptance forms a distinct, stronger commit guard rather than being merged with test success. Syntactically different outputs are aggregated only after their observable behavior has been classified.

\paragraph{Statistics and diagnosis.}
Each pair is reported with quotient total variation, a simultaneous bound from Theorem~\ref{thm:sampling}, and its distinguishing classes. Measuring these quantities at every interface localizes drift to interpretation, generation, validation, repair, or orchestration. Repair results are stratified by failure history, because an unconditional average does not establish Theorem~\ref{thm:repair}'s premise.

\paragraph{Decision rule.}
A replacement is certified for event family $\mathcal A$ only when the composed bound meets its tolerance and every event is quotient-closed. Safe commitment separately requires checker soundness; robustness and correctness are distinct claims.

\section{Implications and Limitations}

A deployable certificate records the observers, model and tool identities, sample counts, confidence level, per-layer defects, composed bound, and commit-guard evidence. It can cover functional behavior, mutation, exceptions, cost, or security, but only observations represented by its equivalences. The PLTS result supports local replacement inside generate--verify--repair loops and communicating agents; guarded safety remains an orthogonal obligation.

The guarantees are conditional on well-chosen observers and proved or measured compatibility. Exact model kernels are unobservable, so empirical audits depend on frozen environments, confidence procedures, and adequate sampling. Nontermination and continuous artifact spaces require measurable quotients or subdistributions beyond the finite presentation developed here. Finally, the controlled audit validates the semantics rather than benchmarking an LLM. A model study should execute the stated protocol with paraphrase families, formalizations, repeated generations, and strong behavioral or proof-based observers; FormalBench and EvalPlus are natural starting points \cite{lecong2025formalbench,liu2023evalplus}.

\section{Conclusion}

\GenOS{} treats an AI coding workflow as a probabilistic program with auditable semantic interfaces. Compatible kernels commute with quotienting and composition; approximate defects form an explicit robustness budget; sound guards establish commit safety. The framework does not assume model invariance. It turns invariance into a local, testable obligation and composes the resulting evidence end to end.

\bibliographystyle{plainnat}
\bibliography{genos_arxiv}

\end{document}